\documentclass[11pt,letterpaper,english]{article}
\usepackage[margin=1in]{geometry}
\usepackage{amsmath,amssymb,amsthm}
\usepackage{xcolor}

\newcommand*{\TCNLA}[1][]{}  

\newcommand{\shortcite}{\cite}

\usepackage[utf8]{inputenc}			
\usepackage{makeidx}
\usepackage{graphicx}
\usepackage{mathtools,amsfonts}
\usepackage{array}
\usepackage{multirow, hhline}

\usepackage{nicefrac}
\usepackage{tikz}
\usetikzlibrary{arrows}
\usetikzlibrary{calc}
\usepackage{paralist}

\usepackage{wrapfig}
\usepackage[english]{babel}		
\usepackage{lipsum}			
\usepackage{comment}
\usepackage{chngpage}			
\usepackage{amsfonts,amssymb,amsmath, amsthm}
\usepackage{amsopn}
\usepackage{color}
\usepackage{pgfplots}
\usepackage{pgfgantt}
\usepackage{algorithm}
\usepackage[noend]{algpseudocode}
\usepackage{subcaption}
\usepackage{tabto}
\usepackage{csquotes} 
\usepackage{multirow}
\usepackage{amsthm}
\usepackage[hidelinks]{hyperref}

\usepackage[draft]{pdfpages}

\usepackage{thmtools} 
\usepackage{thm-restate}

\newcommand{\fact}{\frac{1}{5nm}}
\newcommand{\dEFX}{tEFX}

\newcommand{\abs}[1]{| #1 |}

\newcommand{\argmin}{\text{argmin}}
\newcommand{\payment}{payment}

\algnewcommand{\LineComment}[1]{\State \textcolor{gray}{\texttt{//} \textit{#1}}}

\newtheorem{lemma}{Lemma}

\newtheorem{observation}{Observation}

\newtheorem{definition}{Definition}

\newtheorem{invariant}{Invariant}
\usepackage{xfrac}

\newcommand{\charity}{surplus}
\newcommand{\MPB}{\text{MPB}}

\usepackage{authblk}
\usepackage{hyperref}

\title{Fair and Efficient Allocation of Indivisible Chores with Surplus}

\date{\empty}

\author[1,2]{Hannaneh Akrami}
\affil[1]{Max Planck Institute for Informatics, Germany}
\affil[2]{Graduiertenschule Informatik, Universit\"at des Saarlandes, Germany}
\author[3]{Bhaskar Ray Chaudhury\thanks{Bhaskar Ray Chaudhury was partially supported by NSF Grants CCF-1942321 and CCF-1750436.}}
\affil[3]{University of Illinois at Urbana-Champaign, USA}
\author[3]{Jugal Garg\thanks{Jugal Garg was supported by NSF Grant CCF-1942321.
}}
\author[1]{Kurt Mehlhorn}
\author[3]{Ruta Mehta\thanks{Ruta Mehta was supported by NSF Grant CCF-1750436. 
}}

\begin{document}
\maketitle

\begin{abstract}
    We study fair division of indivisible chores among $n$ agents with additive disutility functions. Two well-studied fairness notions for indivisible items are envy-freeness up to one/any item (EF1/EFX) and the standard notion of economic efficiency is Pareto optimality (PO).  
    There is a noticeable gap between the results known for both EF1 and EFX in the goods and chores settings. The case of chores turns out to be much more challenging. We reduce this gap by providing slightly relaxed versions of the known results on goods for the chores setting. Interestingly, our algorithms run in polynomial time, unlike their analogous versions in the goods setting.  
    
    We introduce the concept of $k$ \charity~
    which means that up to $k$ more chores are allocated to the agents and each of them is a copy of an original chore. We present a polynomial-time algorithm which gives EF1 and PO allocations with $(n-1)$ \charity. 
    
    We relax the notion of EFX slightly and define \dEFX~which requires that the envy from agent $i$ to agent $j$ is removed upon the transfer of any chore from the $i$'s bundle to $j$'s bundle. We give a polynomial-time algorithm that in the chores case for $3$ agents returns an allocation which is either proportional or \dEFX. Note that proportionality is a very strong criterion in the case of indivisible items, and hence both notions we guarantee are desirable. 
\end{abstract}

\section{Introduction}
\label{intro}
Fair division of a set of indivisible items among agents is a fundamental area with applications in various multi-agent settings. The items can be either goods (provides positive utility) or chores (provides negative utility). The case of goods has been vastly studied \cite{survey1}. On the other hand, the case of chores is relatively new. In both settings, given a set $N = [n]$ of $n$ agents and a set $M = [m]$ of $m$ items, the goal is to find an allocation $X = \langle X_1, \ldots, X_n \rangle$ satisfying some fairness and efficiency criteria where agent $i$ receives the bundle $X_i$ for all $i \in [n]$. 

In this paper, we focus on fair division of chores when each agent $i$ has a disutility function 
$d_i: 2^M \rightarrow \mathbb{R}_{\geq 0}$ which indicates how much agent $i$ dislikes each subset $S \subseteq M$ of the chores. We assume that each $d_i$ is additive, i.e., $d_i(S) = \sum_{j\in S} d_i(\{j\})$. 

Envy-freeness is one of the most accepted notions of fairness. 
In the chores setting, allocation $X$ is envy-free if for every pair of agents $i$ and $j$, $d_i(X_i) \leq d_i(X_j)$. However, envy-freeness is too strong to be satisfied.\footnote{For example, consider division of one chore among two agents.} 
Hence, to obtain positive results we need to relax the fairness notion. Therefore, we study envy-freeness up to one item (EF1), envy-freeness up to transferring any item (\dEFX) and proportionality as our fairness criteria. 
For efficiency, we consider (fractional) Pareto optimality (fPO).

\subsection{\boldmath EF1 and fPO with Surplus for $n$ Agents}
An allocation $X = \langle X_1, \ldots, X_n \rangle$ is Pareto optimal (PO), if there exists no allocation $Y = \langle Y_1, \ldots, Y_n \rangle$ such that $d_i(Y_i) \leq d_i(X_i)$ for all agents $i$ and for some agent $j$, $d_j(Y_j) < d_j(X_j)$. 
For Pareto optimality, we assume $Y$ is an integral allocation. 
A stronger notion is fractional Pareto optimality (fPO) which allows $Y$ to be a fractional allocation.
In a fractional allocation $y = \langle y_1, \ldots, y_n \rangle$, $y_{i,c}$ is the fraction of chore $c \in [m]$ allocated to agent $i$ with $\sum_{i \in [n]} y_{i,c} = 1$ and $y_i=(y_{i,1}, \ldots, y_{i,m})$ is $i$'s bundle. Then $d_i(y_i) = \sum_{c \in [m]} y_{i,c} \cdot d_i(\{c\})$ is the disutility of agent $i$ in the fractional allocation.

\paragraph{Fractional Pareto optimality (fPO).} Allocation $x$ is fractionally Pareto optimal or fPO, if there exists no fractional allocation $y$ such that $d_i(y_i) \leq d_i(x_i)$ for all $i$ and for some agent $j$, $d_j(y_j) < d_j(x_j)$.

\paragraph{Envy-freeness up to one chore (EF1).} Allocation $X = \langle X_1, \ldots, X_n \rangle$ is EF1 if for all $i,j \in N$, $d_i(X_i) \leq d_i(X_j)$ or there exists a chore $c \in X_i$ such that $d_i(X_i \setminus \{c\}) \leq d_i(X_j)$. 

EF1 is defined for the case of the goods accordingly, with the difference that the good should be removed from the bundle of the envied agent~\cite{budish11}. For both goods and chores settings, EF1 allocations are known to exist, and they can also be computed in polynomial-time~\cite{LiptonMMS04,efonechore}. 
However, the outputs of these algorithms are not guaranteed to be efficient. Satisfying EF1 and PO simultaneously turns out to be a challenging problem. 

In the goods setting under additive valuations, Caragiannis et al. \shortcite{CaragiannisKMP016} proved that any allocation with maximum Nash welfare is EF1 and PO. Later, Barman et al. \shortcite{BKV18} gave a pseudopolynomial-time algorithm for computing an EF1 and PO allocation, which was recently improved to output an EF1 and fPO allocation~\cite{GargM21}. In the case of chores on the other hand, the existence of EF1 and PO allocations is a big open problem. 
Similar results on chores are known for very limited settings of bivalued disutilities \cite{garg2022fair,ebadian2021fairly}, three agents \cite{chores3} and when chores are divided into two types \cite{limitedchores}.

In this paper, we make progress in this line of work by proving that given additive disutilities, there exists an EF1 and fPO allocation with $(n-1)$ \charity. The analouge of \charity~ in the goods setting is charity, which is a well-accepted concept, and it means that some goods might remain unallocated. Caragiannis et al.~\shortcite{CaragiannisGravin19} introduced the notion of EFX with charity. Many follow-up papers proved relaxations of envy-freeness with charity \cite{CKMS21,BergerCFF,EF2X,mahara2021,laszlo}. 
In the chores setting, by ``$k$ \charity'', we mean that all the chores are allocated, and at most, $k$ extra chores are allocated to the agents, and each of these chores is a copy of an original chore. 

One motivation behind defining the concept of \charity~for chores is the lack of progress on the original problem for over half a decade. It is likely that an allocation that is both EF1 and PO might not always exist, and in that case, the concept of \charity\ seems a good alternative. 

Moreover, duplicating chores makes sense for many applications. For instance, consider the task of distributing papers among reviewers. The goal is to have all papers reviewed and also be fair toward the reviewers. To this end, it does not harm if a few papers are reviewed more than needed. Another practical scenario is when the chores are going to be repeated. Consider the case where the same set of chores needs to be done every month. This can happen in households, corporations, etc. In this case, multiplying some chore $c$ for $k$ times means that we already decide which agents should do $c$ in the following $k$ months. Thus, when planning for the next $k$ months, we can remove $c$ from the set of chores that need to be assigned.

Our first main result is formally stated in Theorem \ref{existence}.
\begin{restatable}{theorem}{existence}\label{existence}
    Given additive disutilites, there exists an allocation with at most $(n-1)$ \charity~ which is EF1 and fPO. Moreover, it can be computed in polynomial time.
\end{restatable}

Note that the allocation in Theorem \ref{existence} being fPO means that it fractionally Pareto dominates all the allocations with the same surplus.
Our approach is based on rounding of competitive equilibrium with equal incomes (CEEI). Since there is no polynomial-time algorithm known for computing a CEEI, we round a $(1-\epsilon)$-approximate-CEEI for $\epsilon=\fact$, which can be computed in polynomial-time~\cite{combinatorial}. By integrally assigning chores which are fractionally allocated in the $(1-\epsilon)$-CEEI, we guarantee that the final allocation is fPO. However, the main challenge here is to achieve EF1 guarantee with at most $n-1$ \charity, which requires careful rounding. 

\subsection{\dEFX~or Proportionality for 3 Agents}
The discrepancy between known results for the goods and chores setting carries over even for instances with a small number of agents. In the goods setting, EFX allocations always exist for $3$ agents with additive utilities \cite{ChaudhuryGM20}. However, the analogous problem for chores is open. An allocation $X = \langle X_1, \ldots, X_n \rangle$ is EFX if for all agents $i$ and $j$ and all chores $c \in X_i$, $d_i(X_i \setminus \{c\}) \leq d_i(X_j)$. The existence of EFX allocations for chores has been studied in the very limited settings of $3$ agents with bivalued disutilites \cite{bivalued3} and also when agents have the same ordinal preferences on the chores \cite{propx}.  

Let us briefly discuss the technique to obtain EFX for three agents for the goods setting and why it fails in the chores setting. In \cite{ChaudhuryGM20}, the high-level idea is to start with an empty allocation and at each step, allocate some unallocated goods to some agents, possibly take away some goods from them or move the bundles among the agents while guaranteeing that the partial allocation is EFX at the end of each step. Basically, the algorithm moves in the space of partial EFX allocations, improving a sophisticated potential function at each step and terminates when it reaches a complete allocation. This algorithm relies on involved concepts such as \emph{champion-graphs} and \emph{half-bundles}. In the goods setting, by allocating more goods, we make progress in the sense of improving agents' utilities. However, in the chores setting, by allocating more chores, we make the agents less happy. Therefore, it is not easy to adapt the algorithm and come up with a potential function which improves after more chores get allocated. In fact, the existence of allocations satisfying even weaker notions of fairness than EFX like \dEFX~is open for the chores setting even when $n=3$. 
Yin and Mehta \shortcite{tEFX3} proved the existence of a \dEFX~allocation for three agents if two of them have additive disutility functions and the ratio of their highest to lowest cost is bounded by two.
\paragraph{Envy-freeness up to transferring any chore (\dEFX).} An allocation is \dEFX~if no agent $i$ envies another agent $j$ after transferring any chore from $i$'s bundle to $j$'s bundle. Formally, allocation $X$ is \dEFX~if for all agents $i$ and $j$ and any chore $c \in X_i$, $d_i(X_i \setminus \{c\}) \leq d_i(X_j \cup \{c\})$. We note that given additive utility/disutility functions, \dEFX~ is stronger than EF2X studied in \cite{EF2X}. EF2X guarantees that any envy is removed upon the removal any two items from the envied/envious bundle.

Recently, Akrami et al. \shortcite{simple} gave an alternative proof for the existence of EFX allocations for three agents in the goods setting which overcomes the mentioned barrier. We use similar techniques, and instead of moving in the space of partial fair allocations and terminating when reaching a complete allocation, we move in the space of complete allocations and stop when we reach a fair allocation. Our technique resembles the cut-and-choose protocol used for fairly allocating items among two agents. In cut and choose, whether the resource is divisible or indivisible, one agent divides it into two parts so that she finds both parts fair. Then the second agent chooses her favorite part and the remaining part goes to the first agent. A similar idea for the case of three agents would be to find a partition $(X_1, X_2, X_3)$ such that agent $1$ finds $X_1$ and $X_2$ fair and agent $2$ finds $X_2$ and $X_3$ fair. This way the third agent can choose her favorite bundle and the remaining bundles can be fairly allocated to the two remaining agents.

An allocation $X$ is proportional if for every agent $i$, $d_i(X_i) \leq d_i(M)/n$. Note that proportionality is too strong to be satisfied when chores are indivisible.\footnote{Again consider the counter example of two agents and one chore.} We show that given any instance comprising of three agents with additive disutilities, in polynomial time one can find an allocation that is either proportional or \dEFX; the choice of alternative is made by the algorithm. Note that the EFX result for $3$ agents in the goods setting is existential and although the approach is constructive, the algorithm is not polynomial. Our second main result is stated in Theorem \ref{EFXmainthmintro2}.
\begin{restatable}{theorem}{EFXmainthmintro2}\label{EFXmainthmintro2}
Given an instance comprising of three agents with additive disutilities, and a set of indivisible chores, there exists an allocation $X$, such that for all $i \in [3]$
\begin{itemize}
    \item either $d_i(X_i) \leq 1/3 \cdot d_i(M)$, or 
    \item for all $c \in X_i$ and $j \in [3]$, we have $d_i(X_i \setminus \{c\}) \leq d_i(X_j \cup \{c\})$.
\end{itemize}
Furthermore, such an allocation can be determined in polynomial time.
\end{restatable}

We remark that although our result does not fully settle the existence of \dEFX~allocations in the chores setting, the guarantees in Theorem~\ref{EFXmainthmintro2} are indeed desirable, especially given that no relaxation of envy-freeness other than EF1, is currently known to exist in the chores setting. Proportionality is a very desirable property of an allocation and is often unattainable in the discrete setting. In fact, the discrete fair division protocol used in Spliddit\footnote{spliddit.org}, 
prior to the Nash-welfare maximization algorithm in 2015\footnote{This is elaborated in Introduction of~\cite{CaragiannisKMP016}.}, first checks for a proportional allocation and only if proportional allocations are unattainable, it attempts at finding relaxations of envy-freeness. There is also research in discrete fair division that attempts to give as many agents their proportional share~\cite{FeigeN22}, whilst satisfying certain relaxations of  classical fairness notions.

\subsection{Further Related Work}
The notion of CEEI has a long history dating back to classical theories in microeconomics~\cite{FisherThesis}. When agents have linear utilities, CEEI with goods is known to be convex, and the equilibrium prices are unique~\cite{EisenbergG59}. Such properties have facilitated the formulation of several polynomial time algorithms~\cite{DPSY08,Orlin10}. In contrast, CEEI with chores forms a non-convex disconnected set~\cite{BogomolnaiaMSY17} and admits several equilibrium prices. Branzei and Sandomirskiy~\shortcite{comp-chores} give a polynomial-time algorithm when the number of agents or the number of goods is constant, which was later improved in~\cite{GargM20,GargHMS21} to the case of mixed manna containing both goods and chores. Later, Chaudhury et al.~\cite{ChaudhuryGMM21} gave a complementary pivot algorithm for finding a CEEI for the case of mixed manna, which runs fast in practice and is provably polynomial-time when the number of agents or the number of items (goods and chores) is constant.  Recently, Boodaghians et al.~\shortcite{BoodaghiansCM22} and Chaudhury et al.~\shortcite{ChaudhuryGMM22} have given polynomial time algorithms for computing $(1-\varepsilon)$-CEEI. However, the complexity of finding an exact CEEI in the chores setting is open. Moreover, Fisher markets that admit integral equilibria is studied in \cite{BarmanKrishna19}.
\section{Preliminaries}
\label{prelim}
An instance of discrete fair division with chores is given by the tuple $\langle N, M, \mathcal {D} \rangle$, where $N=[n]$ is the set of $n$ agents, $M=[m]$ is the set of $m$ indivisible chores and $\mathcal D = (d_1(\cdot), \dots, d_n(\cdot))$, where each $d_i: 2^M \rightarrow \mathbb{R}_{\geq 0}$ is the disutility function of agent $i$. For all agents $i$, $d_i$ is assumed to be \emph{normalized}, i.e., $d_i(\emptyset)=0$ and \emph{monotone}, i.e.,  $d_i(S \cup \{c\}) \geq d_i(S)$ for all $S \subseteq M$ and $c \notin S$. A function $f: 2^M \rightarrow \mathbb{R}_{\geq 0}$ is said to be \emph{additive} if $f(S) = \sum_{s \in S} f(\{s\})$ for all $S \subseteq M$. For ease of notation, we use $c$ instead of $\{c\}$. For $\oplus \in \left\{\leq, \geq ,< , > \right\}$, we use $S \oplus_i T$  for $d_i(S) \oplus d_i(T)$. 

\paragraph{Fisher market.}
In the Fisher market setting for chores in addition to a set $N$ of agents, a set $M$ of chores and a disutility profile $\mathcal{D}$, each agent $i$ has an initial liability $\ell_i > 0$ which specifies how much money this agent should earn in the market. We denote the fisher market instance by $F = \langle N, M, \mathcal{D}, \mathcal{L} \rangle$ where $\mathcal{L} = (\ell_1, \ldots, \ell_n)$. Given the instance $F$, the market outcome is a pair of fractional allocation and payment vector $\langle x, p \rangle$. For all agents $i$ and chores $c$, $x_{i,c}$ denotes what fraction of $c$ is assigned to $i$ and $p_c$ denotes the price of chore $c$. The income of agent $i$ from market outcome $\langle x, p \rangle$ is $p(x_i) = \sum_{c \in M} x_{i, c} p_c$. We can also treat integral bundles as vectors with $0$ and $1$ entries. Given payment vector $p$, the pain per buck of agent $i$ for chore $c$ is $d_i(c) / p_c$. We denote the minimum pain per buck of agent $i$ at payment $p$ by $\MPB_i$, i.e., $\MPB_i = \min_{c \in M} d_i(c) / p_c$. 
\begin{definition}
    Given a Fisher market instance $F$, a market outcome $\langle x, p \rangle$ is a Fisher market equilibrium if 
    \begin{itemize}
        \item the market clears, i.e., for all chores $c \in M$, $\sum_{i \in [n]} x_{i,c} = 1$, and
        \item for all agents $i$, $\sum_{c \in M} x_{i,c} \cdot p_c = \ell_i$, and
        \item all agents only receive chores with minimum pain per buck, i.e., for all agents $i$ and chores $c$, if $x_{i,c} > 0$, then $d_i(c) / p_c = \MPB_i$.
    \end{itemize}
\end{definition}

If for all agents $i$, $\ell_i = 1$, then a Fisher equilibrium is called competitive equilibrium with equal incomes or CEEI. Bogomolnaia et al. \shortcite{BogomolnaiaMSY17} proved that a CEEI always exists when agents have linear disutilities. 

For goods, any Fisher equilibrium is fPO \cite{microeconomics-book}. The same holds true for chores as essentially the same argument shows.
\begin{restatable}
{proposition}{poprop}\label{PO-prop}
    Given additive disutilities, any Fisher equilibrium is fractionally Pareto Optimal. 
\end{restatable}
\begin{proof}
    Let $\langle x,p \rangle$ be a Fisher equilibrium and let $y$ be any other allocation. Then
    \begin{align*}
        \sum_{i \in [n]} \frac{d_i(x_i)}{\MPB_i} &= \sum_{i \in [n]} \sum_{c \in M} x_{i,c} p_c \\
        &= \sum_{c \in M} p_c \\
        &= \sum_{i \in [n]} \sum_{c \in M} y_{i,c} p_c \\
        &\leq \sum_{i \in [n]} \sum_{c \in M} \frac{y_{i,c} d_i(c)}{\MPB_i} \\
        &= \sum_{i \in [n]} \frac{d_i(y_i)}{\MPB_i}.
    \end{align*}
    Hence, it cannot be the case that $d_i(y_i) \leq d_i(x_i)$ for all $i \in [n]$ with one strict inequality.
\end{proof}
Given a market $\langle x, p \rangle$, the \payment~graph of $x$ is a weighted bipartite (undirected) graph with one part consisting of nodes corresponding to the $n$ agents and one part consisting of nodes corresponding to the $m$ chores. We denote the \payment~graph of $x$ by $G_{\langle x, p \rangle}$. There is an edge between agent $i$ and chores $c$, if and only if $x_{i,c} > 0$. For any edge $\{i,c\}$ in $G_{\langle x, p \rangle}$, the weight of $\{i,c\}$ is $e_{i,c} = x_{i,c} \cdot p_c$ which is the earning of agent $i$ from chore $c$ in this market. For any graph $G$, we denote the set of edges of $G$ by $E(G)$.

There is no known polynomial time algorithm for computing a CEEI. However, Boodaghians et al. \shortcite{BoodaghiansCM22} gave an exterior point algorithm to compute a $(1-\epsilon)$-CEEI in polynomial time. The running time was improved by a combinatorial algorithm in \cite{combinatorial}. Namely, a $(1 - \epsilon)$-CEEI can be computed in time polynomial in the size of the input and $\frac{1}{\epsilon}$. In a $(1 - \epsilon)$-CEEI, the income of each agent is between $1-\epsilon$ and $1 + \epsilon$. We formally define $(1 - \epsilon)$-CEEI below. 

\begin{definition}\label{approx-ceei}
    Given a Fisher market $F$, a market outcome $\langle x, p \rangle$ is a $(1 - \epsilon)$-CEEI, for an $\epsilon \in [0, 1]$, if
    \begin{itemize}
    \item the market clears, i.e., for all chores $c \in M$, $\sum_{i \in [n]} x_{i,c} = 1$, and
    \item for all agents $i$, $1-\epsilon \leq \sum_{c \in M} x_{i,c} \cdot p_c \leq 1+\epsilon$, and
    \item all agents only receive chores with minimum pain per buck, i.e., for all agents $i$ and chores $c$, if $x_{i,c} > 0$, then $d_i(c) / p_c = \MPB_i$.
\end{itemize}
\end{definition}

Similar to envy-freeness and its relaxations, we can define payment envy-freeness and its relaxations. In particular, given a payment vector $p=(p_1, \ldots, p_m)$ for the chores, an integral allocation $X$ is \emph{payment envy-free up to one chore} or pEF1, if for all agents $i$ and $j$, either $X_i = \emptyset$ or there exists a chore $c \in X_i$ such that $p(X_i \setminus c) \leq p(X_j)$. 

\begin{restatable}
[Lemma 3.5 in Ebadian \emph{et al.}, 2022]
{proposition}{pefone}\label{pef1-prop}
    If an integral allocation $X$ is pEF1 with respect to payment vector $p$ and $\langle X, p \rangle$ is a Fisher equilibrium, then $X$ is EF1.
\end{restatable}
\begin{proof}
    Consider any two agents $i,j$ such that $X_i \neq \emptyset$. Let $c \in X_i$ be such that $p(X_i \setminus c) \leq p(X_j)$. We have
    \begin{align*}
        X_i \setminus c &=_i \MPB_i \cdot p(X_i \setminus c) &\tag{$\langle X, p\rangle$ is a Fisher equilibrium} \\
        &\leq \MPB_i \cdot p(X_j) &\tag{$X$ is pEF1} \\
        &= \MPB_i \cdot \sum_{c \in X_j} \frac{p(c)}{d_i(c)} \cdot d_i(c) \\
        &\leq_i X_j. &\tag{$\MPB_i = \mathit{min}_{c \in M} d_i(c)/p(c)$}
    \end{align*}
\end{proof}

\section{EF1 + \lowercase{f}PO + Surplus} \label{ef1-sec}

In this section, we prove that after introducing at most $n-1$ chores, an allocation exists which is EF1 and fPO at the same time. Each of these new chores is a copy of an existing chore. Moreover, we compute such an allocation in polynomial time. The high-level idea is to first consider a fractional allocation $x$ which admits a $(1-\epsilon)$-CEEI for $\epsilon = \fact$. Then to each agent, we fully allocate some of the chores that are fractionally allocated to her in $x$. This way, each agent only receives her MPB chores and therefore the allocation is fPO. Furthermore, we guarantee that each agent earns at least $1-\epsilon$ amount of money and there exists a chore that upon its removal, the earned money drops below $1-\epsilon$. This way, we can also guarantee EF1 property for the allocation. In order to achieve such an allocation, we allocate some chores to multiple agents and hence we need multiple copies of some of the chores. However, we prove that the number of required copies does not exceed $n-1$. Basically, our algorithm introduces at most $n-1$ copies of the existing chores and finds an integral Fisher equilibrium where each agent earns $1-\epsilon$ amount of money up to one chore. 

\begin{restatable}{lemma}{acyclic}\label{acyclic-prop}
    Given any Fisher equilibrium $\langle x, p \rangle$ for a Fisher market $F$, there exists a polynomial time algorithm $\mathtt{makeAcyclic}(x,p)$ that computes allocation $y$ such that $\langle y, p \rangle$ is a Fisher equilibrium for $F$ and $G_{\langle y, p \rangle}$ is acyclic. 
\end{restatable}
\begin{proof}
    We define $\mathtt{makeAcyclic}(x,p)$ as following.
    If $G_{\langle x, p \rangle}$ is acyclic then return $\langle x, p \rangle$. Otherwise, as long as $G_{\langle x, p \rangle}$ has a cycle do the following. Let $C = (a_1, c_1 \ldots a_t, c_t, a_1)$ be a cycle in $G_{\langle x, p \rangle}$ where $a_i$ corresponds to the agent nodes and $c_i$ to the chore nodes. Let us denote the earning of agent $i$ from chore $c$ in allocation $x$ as $e^x_{i,c} = x_{i,c} \cdot p_{c}$. Without loss of generality, assume $e^x_{a_1, c_1}$ is minimum among all $e^x_{i,j}$ where $(i,j)$ is an edge in $C$. Now consider the allocation $y$ where for all $i \in [t]$ $e^y_{a_i, c_i} = e^x_{a_i, c_i} - e$, $e^y_{a_i, c_{i-1}} = e^x_{a_i, c_{i-1}} + e$. For all other pairs of $(i,j)$, $e^y_{i,j} = e^x_{i,j}$. In the end of each iteration of detecting a cycle and computing $y$, set $x \leftarrow y$.

    First we prove $\mathtt{makeAcyclic}(x,p)$ terminates in polynomial time. Let $x$ be the allocation in the beginning of each iteration of detecting a cycle and $y$ be the allocation in the end of the iteration. 
    Note that $E(G_{\langle y, p \rangle}) \subsetneq E(G_{\langle x, p \rangle})$ since the edge $(a_1, c_1)$ exists in $G_{\langle x, p \rangle}$ but not in $G_{\langle y, p \rangle}$. Since at each step the number of edges decreases and each step terminates in polynomial time, the procedure terminates in polynomial time and in the end $G_{\langle y, p \rangle}$ is acyclic.
    
    Now we prove the final $\langle y, p \rangle$ is a Fisher equilibrium by induction.
    Note that in the beginning $\langle x, p \rangle$ is a Fisher equilibrium. Now assuming $\langle x, p \rangle$ is a Fisher equilibrium in the beginning of an iteration of removing an edge, we prove in the end of that iteration $\langle y, p \rangle$ is a Fisher equilibrium too. 
    Note that for each chore $c$, $\sum_{i \in [n]} e^y_{i,c} = e^x_{i,c} = p_c$. Thus, all chores are fully allocated in $y$. Also, for each agents $i$, $\sum_{c \in M} e^y_{i,c} = \sum_{c \in M} e^x_{i,c} = \ell_i$. Moreover, for all agents $i$ and chore $c$, if $y_{i,c} > 0$, then $x_{i,c} > 0$. Therefore, in $y$ like in $x$ agents only receive chores with MPB. This means that $\langle y, p \rangle$ is also a Fisher equilibrium. 
\end{proof}
Now we explain Algorithm \ref{existence-alg}.
Given instance $\mathcal{I}$, let $\epsilon = \fact$ and $\langle x,p \rangle =\mathtt{approxCEEI}(\mathcal{I}, \epsilon)$ be the $(1-\epsilon)$-CEEI computed in polynomial time by \cite{combinatorial}. First we run $\mathtt{makeAcyclic}(x,p)$ to make $G_{\langle x,p \rangle}$ acyclic. Then, we compute the integral allocation $Y$ as follows. Our Algorithm consists of two phases. 
We start with $G = G_{\langle x, p \rangle}$ and during Phase $1$, we alter $G$. At each point in time, let $y$ be such that $G$ is the \payment~graph of $\langle y, p \rangle$ (i.e. $G = G_{\langle y, p \rangle}$). Let $N_v$ be the set of the neighbors of node $v$ in $G$. 

\paragraph{Phase 1.} Start from an empty allocation $Y$ and run phase $1$ as long as there is an unallocated chore $c^*$ such that $|N_{c^*}|=1$. Phase $1$ of the algorithm consist of $2$ steps. Basically, as long as there exists an unallocated chore $c^*$ with $|N_{c^*}|=1$, run Step $1$ and then Step $2$.
\paragraph{Step 1.} For all unallocated chores $c$ with $|N_c|=1$, let $i_c$ be the agent such that $N_c=\{i_c\}$. Then add $c$ to $Y_{i_c}$.

\paragraph{Step 2.} For all agents $i$ and chores $c$ such that $\{i,c\} \in E(G)$, if for all chores $c' \in Y_i \cup \{c\}$, $p((Y_i \cup c) \setminus c') > 1 - \epsilon$, then distribute the earning of agent $i$ from chore $c$ equally among the other neighbors of $c$ and remove the edge $\{i,c\}$ from $G$. Recall that $ e_{j,c} = x_{j,c}p_c$ is the earning agent $j$ receives from chore $c$ in the market outcome $\langle x,p \rangle$. Formally, for all $j \in N_c \setminus \{i\}$, we set
    $$e_{j,c} \leftarrow e_{j,c} + \frac{y_{i,c} \cdot p_c}{|N_c|-1}.$$

\paragraph{Phase 2.} The second phase starts when for all unallocated chores $c$, $|N_c| \neq 1$. In Lemma \ref{claim-3} we prove the case $|N_c|=0$ is not possible and therefore for all remaining chores $c$, $|N_c|>1$.
Each of the connected components of $G$ is a tree. For each of the trees $T$ do the following. Take an arbitrary agent $i_0$ in $T$ and consider $T$ rooted at $i_0$. For agent $i_0$, as long as $p(Y_{i_0}) < 1-\epsilon$, keep adding chores from $N_i \setminus Y_{i_0}$ to $Y_{i_0}$. 
Then iterate on the agents in $T$ in a breadth-first order and for each agent $i$ do the following. Let $c_i$ be the chore corresponding to the parent of agent $i$ in $T$. If $c_i$ is not allocated yet, add it to $Y_i$, i.e., $Y_i \leftarrow Y_i \cup \{c_i\}$. Then, keep adding the chores in $N_i \setminus (Y_i \cup \{c_i\})$ to $Y_i$ until $p(Y_i) \geq 1-\epsilon$ or until we run out of chores. Note that all chores in $N_i \setminus (Y_i \cup \{c_i\})$ correspond to children nodes of agent $i$ in $T$. If at the end of this process $p(Y_i) < 1-\epsilon$, add a copy of $c_i$ to $Y_i$.

 Algorithm \ref{existence-alg} shows the pseudocode of our algorithm. In the rest of this section we prove that the final allocation $Y$ is pEF1 and fPO with at most $(n-1)$ \charity. 
\begin{algorithm} [tb]
	\caption{$\mathtt{fairAndEfficient}(\mathcal{I})$ 
        \\ \textbf{Input:} Instance $\mathcal{I}$.
        \\ \textbf{Output:} Allocation $Y$. 
    }
    \label{existence-alg}
	\begin{algorithmic}[1]
	    \State $\epsilon \leftarrow \fact$
	    \State $\langle x, p \rangle \leftarrow$ $\mathtt{makeAcyclic}(\mathtt{approxCEEI}(\mathcal{I}, \epsilon))$
        \State $G \leftarrow$ \payment~graph of $\langle x, p \rangle$ 
        \LineComment{Phase 1:}
        \While{$\exists$ an uncallocated chore $c^*$: $|N_{c^*}|=1$}
	       \LineComment{Step 1:}
            \For {$i \in [n]$} 
	         \State $Y_i \leftarrow \{c \in M | y_{i,c} = 1\}$ 
	       \EndFor
            \LineComment{Step 2:}
            \For {$\{i,c\} \in E(G)$} 
                \If {$\forall c' \in Y_i \cup \{c\}$: $p((Y_i \cup c) \setminus c') > 1 - \epsilon$}
                    \For{$j \in N_c$}
                        \State $e_{j,c} \leftarrow e_{j,c} + \frac{y_{i,c} \cdot p_c}{|N_c|-1}$
                    \EndFor
                    \State $G \leftarrow G \setminus \{\{i,c\}\}$
                \EndIf
            \EndFor
        \EndWhile
        \LineComment{Phase 2:}
	    \For{all connected components $T$ of $G$}
    	    \LineComment{Let $T$ be rooted at $i_0$}
	        \For{all agents $i$ in $T$ in BFS-order}
	            \If {$i \neq i_0$}
	                \State $c_i \leftarrow$ parent chore of $i$ in $T$
    	        \EndIf
	            \If{$c_i$ is not allocated}
	                \State $Y_i \leftarrow Y_i \cup \{c_i\}$
	            \EndIf
    	        \For{$c \in N_i \setminus (Y_i \cup \{c_i\})$}
	                \If {$p(Y_i) < 1-\epsilon$}
	                    \State $Y_i \leftarrow Y_i \cup \{c\}$
	                \EndIf
	            \EndFor
    	        \If {$p(Y_i) < 1-\epsilon$}
	                \State $Y_i \leftarrow Y_i \cup \{c_i\}$
	            \EndIf
	        \EndFor
	    \EndFor
	    \State \Return $Y$
	\end{algorithmic}	
\end{algorithm}
\begin{observation}\label{lowe-p}
    For all agents $i$, $p(y_i) \geq 1-\epsilon$ at any time during Phase $1$.
\end{observation}
\begin{proof}
    The proof is by induction. In the beginning of the algorithm, $y=x$ and thus the claim holds. Now fix an agent $i$ and let $y$ be the allocation such that $G = G_{\langle y, p \rangle}$ before deleting an edge $e$ and $y^*$ be the allocation such that $G \setminus \{e\} = G_{\langle y^*, p \rangle}$. Assuming $p(y_i) \geq 1-\epsilon$, we prove $p(y^*_i) \geq 1-\epsilon$. If $e$ is not incident to $i$, then $p(y^*_i) \geq p(y_i)$ and thus the claim holds. If $e$ is incident to $i$, then $p((Y_i \cup c) \setminus c') > 1-\epsilon$ for all $c' \in Y_i \cup \{c\}$. Therefore, $p(Y_i) = p((Y_i \cup c) \setminus c) > 1-\epsilon$. Note that all chores in $Y_i$ are incident to $i$ in $G_{\langle y^*, p \rangle}$. Therefore, $p(y^*_i) \geq p(Y_i) > 1-\epsilon$.
\end{proof}

\begin{observation}\label{upper-bound-p}
    For all agents $i$, $p(y_i) \leq 1 + (2n-1)\epsilon$ at any time during Phase $1$.
\end{observation}
\begin{proof}
    In the beginning of the algorithm, since $y=x$, we have $\sum_{i \in N} p(y_i) \leq (1+\epsilon)n$. Allocation $y$ changes during Phase $1$ when an edge is deleted in Step $2$. Upon the deletion of edge $\{i,c\}$, $y_{i,c} \cdot p_c$ is distributed among the neighbors of $c$ (in case any such neighbors exist). Therefore, the value of $\sum_{i \in N} p(y_i)$ cannot increase during Phase $1$. Thus, for all agents $i$ at any point during Phase $1$ we have
    \begin{align*}
        (1+\epsilon)n &\geq \sum_{j \in N} p(x_j) \geq \sum_{j \in N} p(y_j) \\
        &\geq (1-\epsilon)(n-1) + p(y_i). &\tag{by Observation \ref{lowe-p}}
    \end{align*}
    Therefore, $p(y_i) \leq 1 + (2n-1)\epsilon$.
\end{proof}
    
\begin{restatable}{lemma}{claimthree}\label{claim-3}
    Before the execution of Phase $2$, for all unallocated chores $c$, $N_c \neq \emptyset$.
\end{restatable}
\begin{proof}
    Towards a contradiction, assume at some point during Phase $1$, $\{i,c\}$ is the only edge incident to $c$ and it gets deleted. Let $y$ be such that $G = G_{\langle y, p \rangle}$ just before deleting $\{i,c\}$. Note that during Phase $1$, as long as a chore has an incident edge, it remains fully allocated. Therefore, $y_{i,c}=1$ since $i$ is the only neighbor of $c$. Also, $p(Y_i) > 1-\epsilon$ (otherwise $\{i,c\}$ would not be deleted). We have
    \begin{align*}
        1 + (2n-1)\epsilon &\geq p(y_i) &\tag{by Observation \ref{upper-bound-p}}\\
        &\geq p(Y_i) + y_{i,c} \cdot p_c \\
        &\geq (1-\epsilon) + p_c. &\tag{by Observation \ref{lowe-p}}
    \end{align*}
    Thus, $p_c \leq 2n\epsilon$. Together with Observation \ref{upper-bound-p} we get 
    \begin{align}
        p(Y_i \cup c) \leq 1 + (2n-1)\epsilon + 2n \epsilon < 1 + 4n\epsilon. \label{ineq-1}    
    \end{align}
    Let $c^*$ be a chore with maximum $p_{c^*}$ in $Y_i \cup \{c\}$. By Pigeonhole principle, $p_{c^*} \geq \frac{p(Y_i \cup c)}{|Y_i \cup \{c\}|} \geq \frac{p(Y_i \cup c)}{m}$. Thus 
    \begin{align*}
        p((Y_i \cup c) \setminus c^*) &\leq \frac{m-1}{m} \cdot p(Y_i \cup c) \\
        &< \frac{m-1}{m} \cdot (1 + 4n\epsilon) &\tag{by Inequality (\ref{ineq-1})}\\
        &\leq 1-\epsilon, &\tag{since $\epsilon = \fact$}
    \end{align*}
    which is a contradiction with the edge $\{i,c\}$ getting deleted. Therefore, for all chores $c$ at least one incident edge of $c$ remains until the end of Phase $1$.
\end{proof}
\begin{observation}\label{complete-aloc}
    All the chores in $M$ are allocated in $Y$.
\end{observation}
\begin{proof}
    By Lemma \ref{claim-3}, in the beginning of Phase $2$ no unallocated chore is isolated in $G$. In Phase $2$, all the chores that are the parent of some agent in $T$ get allocated. Moreover, the leaf chores in $T$ got allocated in Phase $1$. Hence, all the chores are allocated in $Y$.
\end{proof}

\begin{observation} \label{charity}
    The number of copied chores in $Y$ is at most $n-1$.
\end{observation}
\begin{proof}
    In Phase $1$, no chore is allocated more than once. Consider the step in which we allocate chores to agent $i$ when iterating on $T$ in breadth-first order. Note that except $c_i$, all the chores that we allocate to $i$ are her children nodes. Since we run BFS on $T$, these children chores had not been assigned to any other agent before. Therefore, for each non-root agent, we might need to copy one chore and namely her parent node. Thus, the number of copied chores is at most $n-1$.
\end{proof}

\begin{observation}\label{lower-1}
    For all agents $i$, $p(Y_i) \geq 1-\epsilon$.
\end{observation}
\begin{proof}
    Fix an agent $i$. Since $\langle x, p \rangle$ is a $(1-\epsilon)$-CEEI, $p(x_i) \geq 1-\epsilon$. Note that if at some iteration of Step $2$, an adjacent edge of $i$ is deleted, then $p(Y_i) \geq 1-\epsilon$. Now assume no adjacent edge of $i$ is deleted. Let $X_i = \{c \in M | x_{i,c} > 0\}$. We have $p(X_i) \geq p(x_i) \geq 1-\epsilon$. Note that all the chores in $X_i$ which are not added to $Y_i$ in phase $1$ are either children of $i$ in $T$ or her parent node. In either of the cases, as long as $p(Y_i) < 1-\epsilon$, we add these chores to $Y_i$. If we stop before adding the whole chores in $X_i$ to $Y_i$, it means that the condition $p(Y_i) \geq 1-\epsilon$ is satisfied. Otherwise we have $Y_i = X_i$ and thus, $p(Y_i) \geq 1-\epsilon$.
\end{proof}

\begin{observation}\label{upper-1}
    For all agents $i$, there exists a chore $c \in Y_i$ such that $p(Y_i \setminus c) < 1-\epsilon$.
\end{observation}
\begin{proof}
    Consider $Y$ in the end of Phase $1$. By Observation \ref{upper-bound-p}, $p(Y_i) \leq p(y_i) \leq 1 + (2n-1)\epsilon$. Let $c$ be the chore with maximum $p_c$ in $Y_i$. We have
    \begin{align*}
        p(Y_i \setminus c) &\leq \frac{m-1}{m} \cdot p(Y_i) &\tag{$p_c \geq p(Y_i)/m$ by Pigeonhole principle}\\
        &\leq \frac{m-1}{m} \cdot (1 + (2n-1)\epsilon) &\tag{by Observation \ref{upper-bound-p}}\\
        &\leq 1 - \epsilon. &\tag{since $\epsilon = \fact$}
    \end{align*}
    Therefore, there exists a chores $c \in Y_i$, such that $p(Y_i \setminus c) \leq 1-\epsilon$ before the execution of Phase $2$. Also, there exists a chore $c \in Y_i \cup \{c_i\}$ such that $p((Y_i \cup \{c_i\}) \setminus c) < 1-\epsilon$. Otherwise, the edge $(i,c_i)$ would be deleted before Phase $2$. So if in Phase $2$ no chore is added to $Y_i$ or only $c_i$ is added to $Y_i$, the claim holds. Otherwise, let $c$ be the last chore added to $Y_i$. Since we stop adding chores to $Y_i$ the moment $p(Y_i) > 1-\epsilon$, $p(Y_i \setminus c) \leq 1-\epsilon$.
\end{proof}

Now we are ready to prove Theorem \ref{existence}.
\existence*
\begin{proof}
    Let $Y$ be the output of Algorithm \ref{existence-alg}. Let $M'$ be the set of copied chores that are allocated in $Y$ in addition to the chores in $M$. First we prove that $\langle Y, p \rangle$ is a Fisher equilibrium for the market given by $\langle N, M \cup M', \mathcal{D}, (p(Y_1), \ldots, p(Y_n)) \rangle$. By Observation \ref{complete-aloc}, the market clears. Since $\langle x, p \rangle$ is a CEEI for $\langle N, M, \mathcal{D} \rangle$, for each agent $i$, all the chores in $X_i = \{ c \in M | x_{i,c} > 0 \}$ are MPB chores. Since $Y_i \subseteq X_i$, all the chores in $Y_i$ are also MPB chores. In the end, it is clear that each agent $i$ earns $p(Y_i)$. So all the conditions of a Fisher equilibrium hold for $\langle Y, p \rangle$.  Now we prove each of the properties for $Y$ separately.
    \paragraph{EF1.} By Observations \ref{lower-1} and \ref{upper-1}, $Y$ is pEF1. Since $\langle Y, p \rangle$ is a Fisher equilibrium, by Proposition \ref{pef1-prop}, $Y$ is EF1. 
    \paragraph{fPO.} By Proposition \ref{PO-prop}, every Fisher equilibrium is fPO. 
    \paragraph{$\mathbf{(n-1)}$ \charity.} By Observation \ref{complete-aloc}, all the chores in $M$ are allocated and by Observation \ref{charity}, the size of the \charity~is at most $n-1$. 

    Now we prove Algorithm \ref{existence-alg} terminates in polynomial time. The subroutines $\mathtt{makeAcyclic}$ runs in poly$(n,m)$ and $\mathtt{approxCEEI}(x,\epsilon)$ runs in poly$(n,m)$ for $\epsilon = \fact$. Step $1$ can be executed at most $m$ times since in each iteration of Step $1$ a chore gets allocated. Step $2$ can be executed at most $m+n-1$ times since in each iteration of Step $2$ an edge gets deleted. Phase $2$ is a BFS subroutine which terminates in poly$(n,m)$. Therefore, the total running time of Algorithm \ref{existence-alg} is polynomial with respect to $n$ and $m$.
\end{proof}
\paragraph{Remark.} The bound $n-1$ on the size of the surplus is tight for Algorithm \ref{existence-alg}. Consider the instance with $n$ agents and one chore $c$ with disutility $1$ for all the agents. Then any $\epsilon$-CEEI (for $\epsilon = \frac{1}{5n}$) allocates some fraction of $c$ to all of the agents and Algorithm \ref{existence-alg} copies $c$ for $n-1$ times and allocates one copy to each agent.

\section{Fairness Among Three Agents}
\label{Bads-EFX}
Given an allocation $X = \langle X_1, X_2, \dots, X_n \rangle$, we say that an agent $i$ \emph{strongly envies} an agent $j$ if and only if $X_{i} \setminus c >_i X_j \cup c$,  {for some $c \in X_i$}. Thus, an allocation is a \dEFX~allocation if there is no strong envy between any pair of agents. We now introduce certain concepts that will be useful in this section.

\begin{definition}[\dEFX~feasibility]
\label{EFX-feasibility}
Given a partition $X = ( X_1, X_2, \dots, X_n )$ of $M$, a bundle $X_k$ is \dEFX-feasible to agent $i$ if and only if for all chores $c \in X_k$ and all $j \in [n]$,
$$X_k \setminus c \le_i X_j \cup c.$$

Therefore an allocation $X = \langle X_1, X_2, \dots, X_n \rangle$ is \dEFX~if and only if for each agent $i$, $X_i$ is \dEFX-feasible. 
\end{definition}

Note that when agents have additive disutility functions, $X_k$ is \dEFX-feasible for agent $i$ if and only if for all $j \in [n]$,
$X_k \setminus c^* \le_i X_j \cup c^*$ for $c^* = \argmin_{c \in X_k} d_i(c)$.

EFX-feasibility is defined in the same way. Formally, given a partition $X = ( X_1, X_2, \dots, X_n )$ of $M$, a bundle $X_k$ is EFX-feasible to agent $i$ if and only if for all chores $c \in X_k$ and all $j \in [n]$, $X_k \setminus c \le_i X_j.$

Restriction to non-degenerate instances is no loss of generality and simplifies arguments about linear programs. The same is true for allocation of goods and chores. Here, it means that no two distinct bundles of chores are valued the same by any agent.
Chaudhury et al.~\shortcite{ChaudhuryGM20} showed that to prove the existence of EFX allocations in the goods setting, when agents have additive valuations, it suffices to show the existence of EFX allocations for all non-degenerate instances. We adapt their approach and in Appendix \ref{non-degenerate} we 
show that the same claim holds, even when agents have additive disutilities and the notion of fairness is \dEFX.
\textit{Henceforth, in the rest of this section we assume that the given instance is non-degenerate, implying that every agent has positive disutility for every chore.}

In this section we prove Theorem \ref{EFXmainthmintro2}. We start with an allocation which is EFX assuming all agents' disutility functions are $d_1$. During the algorithm we maintain a partition $( X_1, X_2, X_3 )$ of the chores such that all the following invariants hold.
\begin{invariant}\label{inv1}
    $X_1$ and $X_2$ are \dEFX-feasible for agent $1$.
\end{invariant}
\begin{invariant}\label{inv2}
    For all $i \in [2]$ and $c \in X_i$, $X_i \setminus c \leq_1 X_3$.
\end{invariant}
\begin{invariant}\label{inv3}
    $X_3$ is \dEFX-feasible for agent $3$.
\end{invariant}
We use the potential function $\Phi(X) = |X_1| + |X_2|$. Each iteration of our algorithm updates the allocation such that the new allocation is proportional or \dEFX\ or satisfies all the invariants and has a smaller potential value. Since the value of the potential is at most $m$, the number of iterations is at most $m$. 
 
\begin{algorithm} [tb]
	\caption{EFX-Identical} \label{greedy}
	\begin{algorithmic}[1]
	    \State Input : Instance $\mathcal{I} = ([n], M, d)$
	    \State Output: allocation $X$
		\State $X \leftarrow \langle \emptyset, \emptyset, \ldots, \emptyset \rangle$
		\State Let $d(c_1) \geq d(c_2) \geq \ldots \geq d(c_m)$
		\For {$i \leftarrow 1 \text{ to } m$}
		    \State Let $j = \textit{argmin}_{\ell \in [n]} d(X_\ell)$
		    \State $X_j \leftarrow X_j \cup \{c_i\}$
		\EndFor
		\State Return $X$
	\end{algorithmic}	
\end{algorithm}

Li et al. \shortcite{propx} proved when agents have identical ordering on the chores, an EFX allocation can be computed in polynomial time. For completeness, we prove Lemma \ref{greedy-time}.
\begin{restatable}
{lemma}{greedy}\label{greedy-time}
    When all agents have additive disutility function $d$, Algorithm \ref{greedy} returns an EFX allocation in time $\mathcal{O}(m\log m)$.   
\end{restatable}
\begin{proof}
    The proof is by induction on the number of allocated chores. In the beginning, the empty allocation is EFX. Now assume the allocation is EFX right before allocating chore $c_i$ to agent $j$. It suffices to prove that agent $j$ does not strongly envy any other agent. For all chores $c \in X_j \cup \{c_i\}$ and all agents $j' \neq j$ we have
    \begin{align*}
        d(X_j \cup \{c_i\} \setminus \{c\}) &= d(X_j) + d(c_i) - d(c) &\tag{additivity of $d$}\\
        &\leq d(X_j) &\tag{$d(c_i) \leq d(c)$}\\
        &\leq d(X_{j'}). &\tag{$j = \textit{argmin}_{\ell \in [n]} d(X_\ell)$}
    \end{align*}
    Sorting the chores according to their disutility takes $\mathcal{O}(m \log m)$ time. We keep the pairs $(d(X_i), X_i)$ in a priority queue which takes $\mathcal{O}(n \log n)$. Then each round of allocating a chore requires a \textit{delete-min} action ($\mathcal{O}(1)$) and an \textit{insert} ($\mathcal{O}(\log n)$). Hence, Algorithm \ref{greedy} terminates in time $\mathcal{O}(m \log m + n \log n + m \log n) = \mathcal{O}(m \log m)$.
\end{proof}

In the beginning, we run Algorithm \ref{greedy} with $d = d_1$ to obtain allocation $X$. Note that all $X_1$, $X_2$ and $X_3$ are EFX-feasible for agent $1$. Without loss of generality, assume $X_3 \leq_3 X_1 \leq_3 X_2$, i.e., $d_3(X_3) \le d_3(X_1) \le d_3(X_2)$. Then, since all bundles are EFX-feasible for agent $1$, Invariants \ref{inv1} and \ref{inv2} hold and since $X_3$ is the favorite bundle of agent $3$, Invariant \ref{inv3} holds too. If $X_1$ or $X_2$ is \dEFX-feasible for agent $3$, we can allocate a \dEFX-feasible bundle to each of the agents. Without loss of generality assume $X_2$ is also \dEFX-feasible for agent $3$. Then we let agent $2$ pick her favorite bundle. If she picks $X_2$, we assign $X_1$ to agent $1$ and $X_3$ to agent $3$. If agent $2$ picks $X_1$, then we assign $X_2$ to agent $1$ and $X_3$ to agent $3$. The case that agent $2$ picks $X_3$ is symmetric.

Now we assume that $X_3$ is the only \dEFX-feasible bundle for agent $3$. Let $c_1 = \textit{argmin}_{c \in X_1} d_3(c)$. Then the algorithm moves $c_1$ from $X_1$ to $X_3$. Let $X'_1 = X_1 \setminus c_1$, $X'_2 = X_2$ and $X'_3 = X_3 \cup c_1$. The next step of the algorithm depends on whether $X'_2$ is \dEFX-feasible for agent $1$ or not. In Lemma \ref{good-case} we show that if $X'_2$ is \dEFX-feasible for agent $1$ then $X'$ satisfies all the invariants. 
\begin{observation}\label{move-c}
    Let $c_1 = \textit{argmin}_{c \in X_1} d_3(c)$. If $X_3$ is the only \dEFX-feasible bundle for agent $3$ and $X_1 \leq_3 X_2$, then $X_1 \setminus c_1 >_3 X_3 \cup c_1$.
\end{observation}
\begin{proof}
    Assume otherwise. For all $c \in X_1$ we have
    \begin{align*}
        X_1 \setminus c &\leq_3 X_1 \setminus c_1 &\tag{$c_1 \leq_3 c$ and additivity of $d_3$}\\
        &\leq_3 X_3 \cup c_1 \\
        &\leq_3 X_3 \cup c. &\tag{$c_1 \leq_3 c$ and additivity of $d_3$}
    \end{align*}
    Since $X_1 \leq_3 X_2$, $X_1$ is \dEFX-feasible for agent $3$ which is a contradiction. 
\end{proof}

\begin{lemma}\label{good-case}
    If $X'_2$ is \dEFX-feasible for agent $1$, then Invariants \ref{inv1}, \ref{inv2} and \ref{inv3} hold.
\end{lemma}
\begin{proof}
    For all $c \in X'_1$ and $i \in \{2,3\}$ we have
    \begin{align*}
        X'_1 \setminus c &\leq_1 X_1 \setminus c &\tag{$X'_1 \subset X_1$} \\
        &\leq_1 X_i \cup c &\tag{$X_1$ is \dEFX-feasible for agent $1$} \\
        &\leq_1 X'_i \cup c. &\tag{$X_i \subseteq X'_i$}
    \end{align*}
    Therefore, Invariant \ref{inv1} holds. 
    Also, for all $i \in [2]$ and $c \in X'_i$
    \begin{align*}
        X'_i \setminus c &\leq_1 X_i \setminus c &\tag{$X'_i \subseteq X_i$} \\
        &\leq_1 X_3 &\tag{Invariant \ref{inv2} holds for $X$}\\
        &\leq_1 X'_3. &\tag{$X_3 \subset X'_3$}
    \end{align*}
    Thus, Invariant \ref{inv2} holds. By Observation \ref{move-c}, we have $X'_1 >_3 X'_3$. Also, $X'_2 =_3 X_2 \geq_3 X_1 \geq_3 X'_1$. Hence, $X'_3$ is the favorite bundle of agent $3$ and is \dEFX-feasible for her. Therefore, Invariant \ref{inv3} holds as well.
\end{proof}

After moving $c_1$, we have $\Phi(X') = |X'_1|+|X'_2| = |X_1|-1 + |X_2| < \Phi(X)$. Thus, if $X'_2$ is \dEFX-feasible for agent $1$, by Lemma \ref{good-case} all the invariants hold and also the potential function decreases. 

Now we assume that $X'_2$ is not \dEFX-feasible for agent $1$. As long as the second bundle is not \dEFX-feasible for agent $1$, keep moving chores from $X'_2$ to $X'_1$ in non-decreasing order of $d_1(\cdot)$. Formally, let $X'_2 = \{c'_1, c'_2, \ldots, c'_k\}$ and $c'_1 \leq_1 c'_2 \leq_1 \ldots \leq_1 c'_k$. Then $Y_1 = X'_1 \cup \{c'_1, \ldots, c'_\ell\}$ and $Y_2 = X'_2 \setminus \{c'_1, \ldots, c'_\ell\}$ such that $Y_1 <_1 Y_2$ and $Y_1 \cup c'_{\ell+1} \geq_1 Y_2 \setminus c'_{\ell+1}$. Note that $\ell \ge 1$. Let $Y_3 = X'_3$.

\begin{lemma}
    Invariants \ref{inv1} and \ref{inv2} hold for $Y$.
\end{lemma}
\begin{proof}
    We have
    \begin{align*}
        Y_1 &<_1 Y_2 \leq_1 X'_2 \setminus c'_1 &\tag{$Y_2 \subseteq X'_2 \setminus c'_1$}\\
        &\leq_1 X'_3 &\tag{Invariant \ref{inv2} holds for $X'$ by Lemma \ref{good-case}}\\
        &=_1 Y_3 &\tag{$Y_3 = X'_3$}.
    \end{align*}
    Therefore, Invariant \ref{inv2} holds. We also know that for all $c' \in Y_2$, $c' \geq_1 c_{\ell+1}$. Hence, for all $c' \in Y_2$, 
    \begin{align*}
        Y_1 \cup c' \geq_1 Y_1 \cup c'_{\ell+1} 
        \geq_1 Y_2 \setminus c'_{\ell+1} 
        \geq_1 Y_2 \setminus c'.
    \end{align*}
    Since $Y_1 <_1 Y_2$, Invariant \ref{inv1} holds too.
\end{proof}

Now if $Y_3$ is \dEFX-feasible for agent $3$, then all the invariants hold and $\Phi(Y) = |Y_1| + |Y_2| = |X'_1| + |X'_2| = |X_1| + |X_2| -1 < \Phi(X)$. In Section \ref{proportional}, we prove that if $Y_3$ is not \dEFX-feasible for agent $3$, we can obtain a proportional allocation. 

\subsection{Proportional Allocation When $Y_3$ Is Not \dEFX-feasible for Agent $3$}\label{proportional}
In the following observations, we prove that $Y_1$ and $Y_2$ are proportional for agent $1$, and $Y_2$ and $Y_3$ are proportional for agent $3$. Then without any further modification of the bundles, we allocate these bundles to the agents such that the final allocation is proportional.
\begin{observation}\label{prop-3}
    $d_3(Y_3) < d_3(M)/3$.
\end{observation}
\begin{proof}
    We have
    \begin{align*}
        X'_2 =_3 X_2 &\geq_3 X_1 &\tag{$X'_2 = X_2$} \\
        &\geq_3 X'_1 &\tag{$X'_1 = X_1 \setminus c_1$}\\
        &>_3 X'_3 &\tag{Observation \ref{move-c}} \\
        &=_3 Y_3. &\tag{$Y_3 = X'_3$}
    \end{align*}
    Hence, $d_3(X'_3) < d_3(X'_1)$ and $d_3(X'_3) < d_3(X'_2)$. By additivity of $d_3(\cdot)$, we get that $d_3(X'_3) < d_3(M)/3$.
\end{proof}

\begin{observation}\label{prop-2}
    If $Y_3$ is not \dEFX-feasible for agent $3$, then $d_3(Y_2) < d_3(M)/3$.
\end{observation}
\begin{proof}
    We have
    \begin{align*}
        Y_1 &\geq_3 X'_1 &\tag{$X'_1 \subset Y_1$}\\
        &>_3 X'_3 &\tag{Observation \ref{move-c}} \\
        &=_3 Y_3. &\tag{$Y_3 = X'_3$}
    \end{align*}
    Since $Y_3$ is not \dEFX-feasible for agent $3$, it cannot be her favorite bundle. Since $d_3(Y_1) > d_3(Y_3)$, we have $d_3(Y_2) < d_3(Y_3)$. By Observation \ref{prop-3}, $d_3(Y_3) < d_3(M)/3$. Hence, $d_3(Y_2) < d_3(M)/3$.
\end{proof}

Finally, in Observation \ref{prop-1}, we prove that $d_1(Y_1) \leq d_1(M)/3$ and $d_1(Y_2) \leq d_1(M)/3$.
\begin{observation}\label{prop-1}
    $d_1(Y_1) \leq d_1(M)/3$ and $d_1(Y_2) \leq d_1(M)/3$.
\end{observation}
\begin{proof}
    Consider the allocation $\langle X_1 \cup c'_1, X_2 \setminus c'_1, X_3 \rangle$. Note that since Invariants \ref{inv1} and \ref{inv2} hold for $X$, we have $d_1(X_2 \setminus c'_1) \leq d_1(X_1 \cup c'_1)$ and $d_1(X_2 \setminus c'_1) \leq d_1(X_3)$. By additivity of $d_1$, we have $d_1(X_2 \setminus c'_1) \leq d_1(M)/3$. Now it suffices to prove that $d_1(Y_1) \leq d_1(X_2 \setminus c'_1)$ and $d_1(Y_2) \leq d_1(X_2 \setminus c'_1)$. Note that $d_1(Y_1) < d_1(Y_2)$ and $Y_2 \subseteq X_2 \setminus c'_1$. Therefore, we have $d_1(Y_1) < d_1(Y_2) \leq d_1(X_2 \setminus c'_1)$.
\end{proof}

At this stage of the algorithm, by Observation \ref{prop-1} we have that $d_1(Y_1) \leq d_1(M)/3$ and $d_1(Y_2) \leq d_1(M)/3$. Also by Observations \ref{prop-3} and \ref{prop-2}, we have $d_3(Y_3) < d_3(M)/3$ and $d_3(Y_2) < d_3(M)/3$. Now we let agent $2$ pick her favorite bundle. Let it be $Y_i$. Clearly, $d_2(Y_i) \leq d_2(M)/3$. As already argued before, no matter which bundle agent $2$ chooses, we can allocate one of $Y_1$ or $Y_2$ to agent $1$ and one of $Y_2$ or $Y_3$ to agent $3$. Therefore, we obtain a proportional allocation.

\section{Conclusion}
We have introduced a concept of $k$ \charity\ and showed the existence of an allocation that is both EF1 and fPO with at most $n-1$ \charity\ in the case of indivisible chores. Furthermore, such an allocation can be computed in polynomial time. A natural open question is whether there exists an allocation that is both EF1 and fPO with $< n-1$ \charity. 

Our second result shows the existence of allocation of indivisible chores among 3 agents that is either \dEFX or proportional. Since proportionality is a very strong guarantee, which is not possible to satisfy for every instance, this result is the first non-trivial result for a slight relaxation of EFX for 3 agents. A natural open question is whether EFX allocations exist for 3 agents.

\bibliographystyle{plainurl}
\bibliography{sample}

\appendix
\section{Non-Degenerate Instances}\label{non-degenerate}
We call an instance $I = \langle [n], M, \mathcal{D} \rangle$ non-degenerate if and only if no agent dislikes two different sets equally, i.e., $\forall i \in [n]$ we have $S \neq_i T$ for all $S \neq T$. We adapt the technique in~\cite{ChaudhuryGM20} and  show that it suffices to deal with non-degenerate instances when there are $n$ agents with additive disutilities, i.e., if there exists a \dEFX~allocation in all non-degenerate instances, then there exists a \dEFX~allocation in all instances. 

Let $M = \left\{c_1, c_2, \dots, c_m \right\}$. We perturb any instance $\mathcal{I} = \langle [n],M ,\mathcal{D} \rangle$ to $\mathcal{I}(\epsilon) = \langle [n],M ,\mathcal{D}(\epsilon) \rangle$, where for every $d_i \in \mathcal{D}$ we define $d'_i \in \mathcal{D}(\epsilon)$, as

$$ d'_i(c) = d_i(c_j) + \epsilon \cdot 2^{j} \quad \quad \forall S \subseteq M$$ 

\begin{lemma}
	\label{non-degeneracy-technical} \label{non-degeneracy-main}
	Let $$\delta = \min_{i \in [n]} \min_{S,T \colon d_i(S) \neq d_i(T)} \abs{ d_i(S) - d_i(T)}$$ 
	and let $\epsilon > 0$ be such that $\epsilon \cdot 2^{m+1}  < \delta$. Then
	\begin{enumerate}
		\item For any agent $i$ and $S,T \subseteq M$ such that $d_i(S) > d_i(T)$, we have $d'_i(S) > d'_i(T)$.
		\item $\mathcal{I}(\epsilon)$ is a non-degenerate instance. Furthermore, if $X = \langle X_1, \ldots ,X_n \rangle$ is a \dEFX~allocation for $\mathcal{I}(\epsilon)$ then $X$ is also a \dEFX~allocation for $\mathcal{I}$.
	\end{enumerate}
\end{lemma}
\begin{proof}
	For the first statement of the lemma, observe that 
	\begin{align*}
	d'_i(S) - d'_i(T)  &=  d_i(S) - d_i(T)  + \epsilon(\sum_{c_j \in S \setminus T}2^j - \sum_{c_j \in T \setminus S}2^j) \\
	&\geq \delta -  \epsilon \sum_{c_j \in T \setminus S}2^j\\
	&\geq \delta -  \epsilon \cdot (2^{m+1}-1)\\
	&>0 \enspace .  
	\end{align*}
	
	For the second statement of the lemma, consider any two sets $S,T \subseteq M$ such that $S \neq T$. Now, for any $i \in [n]$, if $d_i(S) \neq d_i(T)$, we have $d'_i(S) \neq d'_i(T)$ by the first statement of the lemma. If $d_i(S) = d_i(T)$, we have $d'_i(S) - d'_i(T) = \epsilon(\sum_{c_j \in S \setminus T}2^j - \sum_{c_j \in T \setminus S}2^j) \neq 0$ (as $S \neq T$). Therefore, $\mathcal{I}(\epsilon)$ is non-degenerate.
	
	For the final claim, let us assume that $X$ is a \dEFX~allocation in $\mathcal{I}(\epsilon)$ and not a \dEFX~allocation in $\mathcal{I}$. Then there exist $i,j$, and $c \in X_i$ such that $d_i(X_i \setminus c) > d_i(X_j \cup c)$. In that case, we have $d'_i(X_i \setminus c) > d'_i(X_j \cup c)$ by the first statement of the lemma, implying that $X$ is not a \dEFX~allocation in $\mathcal{I}(\epsilon)$ either, which is a contradiction. 
\end{proof}

\end{document}